\documentclass[12pt]{article}
\usepackage{amssymb}%,natbib}
\usepackage{graphicx}
\usepackage{amsmath}
\usepackage{amsfonts,amsthm}
\usepackage{natbib}
\usepackage{subcaption}
\newcommand{\mbi}[1]{\mbox{\boldmath $ 1$}}

\newtheorem{proposition}{Proposition}

\newcommand{\bm}[1]{\mbox{\boldmath $#1$}}
\newcommand{\argmin}{\mathop{\rm argmin}}

\newcommand{\bms}[1]{\mbox{\scriptsize\boldmath $#1$}}
\title{Information Geometry for Maximum Diversity Distributions}
\author{Shinto Eguchi}
\date{\today}
\begin{document}
\setcitestyle{authoryear,open={(},close={)}}
%\bibliographystyle{jae.bst}
%\bibliographystyle{plainnat.bst}
%\centerline{ \Large The power cross-entropy for multilayer perceptron}
\maketitle
\begin{abstract}
{In recent years, biodiversity measures have gained prominence as essential tools for ecological and environmental assessments, particularly in the context of increasingly complex and large-scale datasets. 
	We provide a comprehensive review of diversity measures, including the Gini-Simpson index, Hill numbers, and Rao's quadratic entropy, examining their roles in capturing various aspects of biodiversity. Among these, Rao's quadratic entropy stands out for its ability to incorporate not only species abundance but also functional and genetic dissimilarities.
	The paper emphasizes the statistical and ecological significance of Rao's quadratic entropy under the information geometry framework. %, focusing on the maximum diversity distribution under constraints  
	We explore the distribution maximizing such a diversity measure under linear constraints that reflect ecological realities, such as resource competition or habitat suitability.
	Furthermore, we discuss a unified approach of the Leinster-Cobbold index combining Hill numbers and Rao's entropy, allowing for an adaptable and similarity-sensitive measure of biodiversity. 
	%We argue that this approach can guide conservation efforts by highlighting optimal diversity configurations in response to ecological pressures, providing insights into community resilience, species dominance, and ecosystem stability. 
	Finally, we discuss the information geometry associated with the maximum diversity distribution focusing on the cross diversity measures such as the cross-entropy.}

\end{abstract}

\newpage
\section{Introduction}
This paper examines the application of information geometry to biodiversity measurement. 
Various indices, including the Simpson-Gini index, Hill numbers, and Rao's quadratic entropy, have been developed and utilized to assess the ecological states of biological communities.
In particular, Rao's quadratic entropy has become a cornerstone of biodiversity research , which incorporates both species abundance and functional or genetic dissimilarity, see \cite{rao1982a,rao1982b}.
By combining species abundance with functional or genetic differences, Rao's quadratic entropy offers a comprehensive tool for ecological and population genetics studies, enabling ecologists to measure and interpret diversity in ways that were not possible before. 
By combining species abundance with functional and genetic differences across large datasets, it provides insights into ecosystem stability, resilience, and conservation needs. As computational power and data accessibility continue to grow, Rao's quadratic entropy will likely remain at the forefront of biodiversity assessments, driving advances in ecological understanding and environmental stewardship.
Rao's groundbreaking contributions in statistical theory and information geometry have had a profound influence on the development of diversity measures, see \cite{rao1945}. 
His introduction of the Fisher-Rao information metric laid the foundation for viewing statistical models as geometric entities, an approach that has become instrumental in machine learning, ecological modeling, and data science, see \cite{{rao1961},{rao1987differential}}.

We present a comprehensive review of diversity indices through the lens of information geometry. 
%We investigate the maximum diversity distribution under a linear constraint, 
This approach highlights the deep interconnections between Rao's information-theoretic measures and ecological resilience, revealing how the geometric properties of diversity measures can enhance our understanding of community structure.
Specifically, we discuss the maximum diversity distributions under linear constraints, examining how ecological factors, such as resource limitations or habitat suitability, can be represented as geometric constraints on species distributions. 
We reduce the problem of finding the maximal distribution to a mathematical optimization characterized by the Karush-Kuhn-Tucker (KKT) conditions to find an optimal solution.
%drawing parallels between the maximum entropy distributions and the maximum divergence distributions from an information geometric perspective.
We demonstrate that the maximum Hill-number distributions can be characterized using \( q \)-geodesics, which generalize the concept of geodesic paths in the simplex and provide insightful connections between statistical estimation and ecological interpretations.
This approach extends Rao's vision of information geometry, providing ecologists and conservationists with a robust method for identifying optimal diversity configurations in response to environmental constraints. 
In addition, we propose practical applications for this framework, emphasizing how Rao's contributions to information geometry continue to shape and advance biodiversity research, guiding conservation strategies and promoting a deeper understanding of ecosystem dynamics.
From this information-geometric viewpoint, we offer valuable insights into ecosystem stability and resilience. Our approach not only advances theoretical understanding but also has practical implications for biodiversity conservation and environmental management.

The paper is organized as follows:
Section \ref{sect2} builds a mathematical framework for diversity measures, in which we overview the standard indices of diversity such as the Gini-Simpson index, the Hill number and Rao's quadratic entropy.
In Section \ref{sect3} we present an idea of maximum divergence distributions under a linear constraint focusing on the Hill numbers.  The approach is extended to a case of multiple constraints. The ecological interpretation is explored from the practical point of view. 
Section \ref{sect4} gives information-geometric understanding for the maximum divergence distributions
in a simplex, or a space of categorical distributions.  
Notably, we associate a foliation, a partitioning into submanifolds, of the simplex with the maximum divergence distribution under linear constraints. This geometric structure provides deeper insight into the behavior of diversity measures within the simplex of categorical distributions.
In Section \ref{sect5} we discuss the maximum divergence distribution in comprehensive perspectives.

%%%%%%%%%%%%%%%%%%%%%%%%%%%%%%%%%%%%%%%%%%%%%%%%%%%%%%%%%%%%%%%%%%%%%%%%%%%%%%%%%%%%%%%%%%%%%%%%%%%%%%%%%%%%%%%%%%%%%%%%%%%%%%%%%%%%%%%%%%%%%%%%%%%%%%%%%%%%%%%%%%%%%%%%%%%%%%%%%%%%%%%%%%%%%%%%%%%%%%%%%%%%%%%%%%%%%%%%%%%%%%%%%%%%%%%%%%%%%%%%%%%%%%%%%

\section{Biodiversity measures}\label{sect2}

We provide an overview of common measures used to quantify the diversity of a community, cf. \cite{may1975}. 
Consider a community with \( S \) species, where \( p_i \) represents the relative abundance of species \( i \). The set of all possible relative abundance vectors is represented by the \( (S-1) \)-dimensional simplex:
\[
\Delta_{S-1} = \Big\{ \bm{p} = (p_1, \ldots, p_S) : \sum_{i=1}^S p_i = 1, \, p_i \geq 0 \, (i = 1, \ldots, S) \Big\}.
\]
Henceforth, we will identify \( \Delta_{S-1} \) with the set of all \( S \)-variate categorical distributions. 
The following measures of diversity, applied in various ecological contexts, enable various views of biodiversity, emphasizing different aspects from species evenness and dominance.
The Gini-Simpson index  is given by
\[
D_\text{GS}({\bm p}) = 1 - \sum_{i=1}^S p_i^2
\]
for ${\bm p}=(p_i)$ of $\Delta_{S-1}$, see \cite{simpson1949}.
The Gini-Simpson index represents the probability that two randomly selected individuals from a community belong to different species. This index is sensitive to species evenness and increases as species abundances become more evenly distributed.
The Hill Numbers provide a general framework for diversity that is sensitive to the order \( q \), which adjusts the emphasis on species richness versus evenness. The Hill number of order \( q \) is defined as:
\begin{align}\label{Hill}
	{}^q\!D({\bm p}) = \left( \sum_{i=1}^S p_i{}^q\! \right)^{\frac{1}{1 - q}}
\end{align}
where \( q \) is a controlling parameter that determines the sensitivity of the index to species abundances, see \cite{hill1973}.
When \( q = 0 \), \( {}^0\!D({\bm p}) = S \), which is simply species richness (the count of species).
When \( q \) goes to $1$, the measure becomes Boltzmann-Shannon entropy, with \( {}^q\!D({\bm p}) \) calculated as
\(
{}^1\!D({\bm p}) = \exp\left(H({\bm p}) \right),
\)
where
\[
H({\bm p}) = -\sum_{i=1}^S p_i \log(p_i).
\]
This form uses the exponential of Boltzmann-Shannon entropy  $H({\bm p})$ to remain in the Hill framework.
This quantifies the uncertainty in predicting the species of a randomly chosen individual. It is maximized when all species are equally abundant, indicating high diversity, and decreases as dominance by fewer species increases.
When \( q = 2 \), The measure reduces to the inverse Simpson index
\[
{}^2\!D({\bm p}) = \left( \sum_{i=1}^S p_i^2 \right)^{-1}
\]
which gives greater weight to common species and is particularly useful for measuring dominance.
Thus, the parameter \( q \) adjusts the sensitivity of the Hill number to species abundances, with higher \( q \)-values emphasizing dominant species and lower values favoring rare species.
The Berger-Parker index focuses on dominance by measuring the proportional abundance of the most abundant species
\[
D_{\text{BP}}({\bm p}) = \max(p_i)
\]
where \( p_i \) is the relative abundance of each species. The Berger-Parker index is useful for identifying communities dominated by a single or few species, with higher values indicating lower diversity and higher dominance by certain species.
Alternatively, the inverse Berger-Parker  index  \( 1/D_{\text BP}({\bm p}) \) can be used as a diversity measure, where higher values indicate greater diversity.
When the order of the Hill number goes to infinity, it goes the inverse Berger-Parker index. This is because
\[
\lim_{q\rightarrow\infty} {}^q\!D({\bm p}) = \lim_{q\rightarrow\infty} \left[\max(p_i)
\left\{ \sum_{i=1}^S \left(\frac{p_i}{\max(p_i)} \right)^q\right\}^{\frac{1}{q}}\right]^\frac{q}{1-q}
\]
which is nothing but \( 1/D_{\text BP}({\bm p}) \). Alternatively, as $q\rightarrow-\infty$, the Hill number converges to the inverse of the smallest relative abundance: $\lim_{q\rightarrow-\infty} {}^q\!D({\bm p})=1/\min_i(p_i)$.

Rao' s quadratic entropy is a measure of diversity that accounts not only for species abundance but also for the dissimilarity between species, such as functional or genetic differences, see \cite{rao1982a,rao1982b}. 
It is defined as
\[
{\mathrm Q}({\bm p}) = \sum_{i=1}^S \sum_{j=1}^S p_i p_j {W}_{ij},
\]
where \(\bm p=(p_1,...,p_S) \) represents the relative abundances, and \(W_{ij} \) is a measure of dissimilarity between species \( i \) and \( j \) satisfying $W_{ii}=0$ and $W_{ij}\geq0$.
For example, the dissimilarity matrix is often derived from a distance matrix calculated using measures such as Gower's distance, Manhattan distance, and others, see \cite{{botta-dukat2005},{southwood2009}}.
Rao' s quadratic entropy increases with both the abundance of different species and their functional or genetic differences. 

When \( W_{ij} = 1 \) for all \( i \neq j \), Rao's quadratic entropy reduces to a form of the Gini-Simpson  index. 
This measure is particularly useful for quantifying functional diversity in ecosystems where species differ in ecological roles or genetic traits.
%The relationship between the Hill number and Rao's quadratic entropy in the diversity index proposed by Leinster and Cobbold is central to their unified approach to biodiversity measurement.
%''' 1. Hill Number as Effective Diversity Measure
The Hill number \( {}^q\!D({\bm p}) \) is a family of diversity indices parameterized by \( q \), which controls the sensitivity of the measure to species abundances. 
Hill numbers are `effective' numbers, meaning they give an intuitively meaningful measure of the effective number of species in a community by reflecting both species abundances and the emphasis on common versus rare species through the parameter \( q \).
Rao's quadratic entropy, on the other hand, incorporates both species abundances and a dissimilarity matrix that quantifies pairwise similarities between species, allowing it to account for the phylogenetic or functional differences between them. Mathematically, it is defined as the expected dissimilarity between two randomly chosen individuals from a community.
In essence, Rao's entropy provides a diversity measure that is sensitive not only to species abundances but also to their similarity, making it valuable for communities where species may vary significantly in terms of genetic, functional, or ecological traits.
%Leinster and Cobbold's framework integrates
To integrate these two concepts the Leinster-Cobbold index is proposed by
\[
{}^q\!D_{\bms Z}({\bm p}) = \left( \sum_{i=1}^S p_i({\bm Z}{\bm p})_i^{q-1} \right)^{\frac{1}{1 - q}}
\]
where ${\bm Z}$ is the similarity matrix with entries $Z_{ij}$ representing the similarity between species $i$ and $j$ satisfying $Z_{ii}=1$ and $0\leq Z_{ij}\leq1$, see \cite{leinster-cobbold2012}.
The term $(\bm Z \bm p)_i=\sum_{j=1}^S Z_{ij}p_j$ is the average similarity of species $i$ to all species in the community, weighted their relative abundances.
This retains the effective number interpretation of Hill numbers while incorporating a similarity matrix like Rao's quadratic entropy. 
For a given community, they define a generalized diversity measure that combines both relative abundances and similarities among species. 
The sensitivity parameter \( q \) of the Hill number framework is retained, controlling the emphasis on rare versus common species, and a similarity matrix \( {\bm Z} \) quantifies species resemblance.
Leinster and Cobbold's index \( {}^2\!D_{\bms Z}({\bm p}) \) with $q=2$ simplifies
\[
{}^2\!D_{\bms Z}(\bm p)=\frac{1}{\bm p^\top \bm Z\bm p},
\]
which, under some conditions, aligns with the inverse of Rao's quadratic entropy when $\bm Z={\bf I}-\bm W$ (i.e., the similarity matrix is derived from the dissimilarity matrx).
This makes it a direct similarity-sensitive generalization of the Gini-Simpson index.
By varying \( q \), one obtains a diversity profile that reflects different levels of sensitivity to rare and common species, while the similarity matrix \( {\bm Z} \) ensures that closely related species contribute less to the overall diversity measure. This connection allows both Hill numbers and Rao's entropy to be seen as part of a single framework, with \( {}^q\!D_{\bms Z}({\bm p}) \) converging to the traditional Hill numbers when species are maximally distinct (i.e., the similarity matrix \( {\bm Z} \) is an identity matrix).
In this unified measure, the flexibility of Hill numbers in emphasizing different aspects of species abundance combines with Rao's consideration of species dissimilarity, providing a comprehensive and adaptable biodiversity measure that bridges the strengths of both approaches.

%%%%%%%%%%%%%%%%%%%%%%%%%%%%%%%%%%%%%%%%%%%%%%%%%%%%%%%%%%%%%%%%%%%%%%%%%%%%%%%%%%%%%%%%%%%%%%%%%%%%%%%%%%%%%%%%%%%%%%%%%%%%%%%%%%%%%%%%%%%%%%%%%%%%%%%%%%%%%%%%%%%%%%%%%%%%%%%%%%%%%%%%%%%%%%%%%%%%%%%%%%%%%%%%%%%%%%%%%%%%%%%%%%%%%%%%%%%%%%%%%%%%%%%%%

\section{Maximum diversity distributions}\label{sect3}
This section explores the mathematical foundations of maximizing biodiversity measures, focusing on Hill numbers, Rao's quadratic entropy, and the Leinster-Cobbold index. 
We present optimization results under realistic ecological constraints, such as linear resource limits, and examine their implications for understanding community dynamics. 
By integrating these mathematical insights with ecological interpretations, we aim to provide a robust framework for evaluating and managing biodiversity in both theoretical and applied contexts.

Hill numbers ${}^q\!D({\bm p})$are interpreted as effective species counts, emphasizing their bounded nature and ensuring that the do not exceed the actual species count $S$, see \cite{{hill1973},{jost2006}}.
Here, we confirm the maximum of the Hill numbers within a mathematical framework.

\

\begin{proposition}\label{prop1}
	Let ${\bm e}=(e_i)$ be a uniform distribution, that is, $e_i=\frac{1}{S}$ for $i=1,...,S$.
	Then, the Hill number ${}^q\!D({\bm p})$ attains its maximum value $S$ at ${\bm e}$.
\end{proposition}
\begin{proof}
	We aim to show that ${}^q\!D({\bm p})\leq S$ for all $\bm p\in\Delta_{S-1}$, with equality if and only if $\bm p=\bm e$.\\
	Case 1: $q > 1$. 
	The function $f(p)=p^q$ is convex on $[0,1]$ when $q>1$.
	By Jensen's inequality for convex functions:
	\begin{align}\label{conclude}
		\frac{1}{S}\sum_{i=1}^S p_i{}^q \geq\Big( \frac{1}{S}\sum_{i=1}^S{p_i}\Big)^q=S^{-q}.
	\end{align}
	Thus, ${}^q\!D({\bm p})\leq S$. The equality holds when and only when ${\bm p}={\bm e}$.\\
	Case 2: $q<1$, The function $f(p)=p^q$ is concave on $[0,1]$ when $q<1$.
	The reverse of inequality \eqref{conclude} holds and hence, $\sum_{i=1}^S p_i{}^q\! \leq S^{1-q}$, or we get the same inequality: ${}^q\!D({\bm p})\leq S$ due to $1-q>0$. Similarly, the equality holds when and only when ${\bm p}={\bm e}$.
	Thus, in both cases, the maximum Hill number is $S$ and is attained uniquely at the uniform distribution $\bm e$.
\end{proof}
Extending this result, over a \((K-1)\)-dimensional simplex $\Delta_{K-1}$ (i.e., considering only $K$ out of $S$ species),  the Hill number \({}^q\!D({\bm p})\) attains its maximum value \(K\)  when  \({\bm p}\) is the uniform distribution over those  \(K\) species, with \(p_i = \frac{1}{K}\) for the $K$ species and \(p_i = 0\) for the remaining \(S-K\) species.
In this way, the Hill number satisfies $1\leq{}^q\!D({\bm p})\leq S$. When ${}^q\!D({\bm p})=1$,  the community is entirely dominated by a single species, indicating. 
Conversely, when ${}^q\!D({\bm p})=S$, the community has maximum diversity with species occurring in equal abundances.  
Therefore, the Hill number ${}^q\!D({\bm p})$ represents
the effective number of species, providing a meaningful measure of diversity that accounts for both species richness and evenness. 
It is helpful to compare values of the Hill number across different ecosystems or treatments, considering both the index value and how close it is to its theoretical maximum for that number of species, see \cite{chao1987} for ecological perspectives.

However, in real ecosystems, perfect evenness is rare, and the maximizing may not always be ecologically desirable or feasible, see \cite{colwell-coddington1994}.
Maximizing the Hill number might conflict with other ecological goals or constraints.
To address this,  we introduce a linear constraint on species distribution to reflect real-world ecological limits like resource availability, climate tolerance, and habitat capacity.
We fix a linear constraint
\[
\sum_{i=1}^S a_i p_i = C,
\]
where $a_i$'s and $C$ are fixed constants.
Here, without losing generality, we assume $a_i$ is positive.
The constant $C$ reflects a fixed ecological limit, which lies within the range $a_{\rm min}\leq C\leq a_{\rm max}$ to ensure the existence for feasible $\bm{p}=(p_i)$, where $a_{\min}=\min_{i}a_i$ and $a_{\max}=\max_{i}a_i$.
In practice, the value of \( C \) is often set as
\(
C = \sum_{i=1}^S a_i \hat\pi_i,
\)
where \( \hat\pi_i \) represents an initial estimate of species proportions based on preliminary field data or other relevant ecological information. This approach anchors \( C \) in observed or estimated ecological conditions, providing a realistic basis for the constraint on species distribution.
This can represent ecological factors such as resource preference, habitat suitability, or species-specific traits.
For instance, in a desert ecosystem, \( a_i \) might represent drought tolerance in plants, favoring species with high \( a_i \) values as they better adapt to dry conditions.
%In a polluted ecosystem, \( a_i \) might represent pollution tolerance, leading to higher \( p_i \) values for pollution-resistant species.
The distribution can reflect species' access to shared resources under competition. Species with high \( a_i \) values may beat  others in resource acquisition, resulting in larger \( p_i \) values. 

In this context, the linear constraint model is particularly relevant for habitats where specific resources are scarce and heavily contested. Such constraints shape community structures by influencing which species thrive, based on factors such as resource efficiency or adaptability to environmental conditions.
For a general value of $q$ in the Hill number framework, the maximum Hill number distribution under a linear constraint is characterized by a power-law form of $q$  as follows.

\

\begin{proposition}\label{prop2}
	Assume a linear constraint $\sum_{i=1}^Sa_i p_i =C$, where $a_i$'s and $C$ are constants.
	\begin{enumerate} 
		\item For $q<1$,  the Hill number \({}^q\!D(\bm p)\) is maximized at
		\begin{align}\label{hat-p}
			\hat p_{\theta i} = \frac{ (1+(q-1)\theta a_i)^{\frac{1}{q - 1}}}{\sum_{j=1}^S (1+(q-1)\theta a_j)^{\frac{1}{q - 1}}} \quad(i=1,...,S),
		\end{align}
		if $C$ satisfies $ C_{\bms a} \leq C \leq a_{\max}$, where $\theta$ is determined by the constraint.  
		Here
		\begin{align}\nonumber %\label{hat-p}
			C_{\bms a}= \frac{ \sum_{i=1} a_i^{\frac{q}{q-1}}}{\sum_{i=1} a_i^{\frac{1}{q-1}}} \quad(i=1,...,d).
		\end{align}
		\item For $q>1$,  the Hill number \({}^q\!D(\bm p)\) is maximized at
		\begin{align}\label{hat-p2}
			\hat p_{\theta i} = \frac{ [1+(q-1)\theta a_i]_+^{\frac{1}{q - 1}}}{\sum_{j=1}^S [1+(q-1)\theta a_j]_+^{\frac{1}{q - 1}}} \quad(i=1,...,S),
		\end{align}
		if $C$ satisfies $ a_{\min} \leq C \leq C_{\bms a}$, where $[x ]_+=\max\{x,0\}$, $\theta$ is determined by the constraint.
	\end{enumerate}  
\end{proposition}
\begin{proof}
	Our goal is to maximize \( {}^q\!D({\bm p}) \) under the constraints \( \sum_{i=1}^S p_i = 1 \) and  \( \sum_{i=1}^S a_i p_i =C \).\\
	Case 1: $q<1$. Since  \( {}^q\!D({\bm p}) \) is an increasing function of $\sum_{i=1}^S p_i^q$ when $q<1$, maximizing  \( {}^q\!D({\bm p}) \) is equivalent to maximizing  $\sum_{i=1}^S p_i^q$.
	We construct the Lagrangian: 
	\[
	\mathcal{L} = \frac{1}{q}\sum_{i=1}^S p_i{}^q - \lambda \Big( \sum_{i=1}^S a_i p_i -C \Big) - \mu \Big( \sum_{i=1}^S p_i - 1 \Big) - \sum_{i=1}^S \nu_i p_i,
	\]
	where  \( \mu \), \( \lambda \), and $\nu_i\geq0$ are Lagrange multipliers.
	Taking the derivative of \( \mathcal{L} \) with respect to \( p_i \) and setting it to zero gives:
	\[
	\frac{\partial \mathcal{L}}{\partial p_i} =  p_i^{q - 1} - \lambda a_i - \mu - \nu_i= 0
	\]
	If $p_i>0$, then $\nu_i=0$, and thus $p_i^{q-1}\propto \mu+\lambda a_i$.
	If $p_i=0$,  then $p_i^\frac{1}{q-1}$ is undefied because of $q<1$, so all $p_i$ must be positive.
	Setting $\theta=\lambda/\mu(q-1)$, we obtain the global maximizer $\hat p_{\theta i}$ in \eqref{hat-p}
	% for $\theta \geq -1/a_{\rm max}$.
	since $( \hat p_{\theta i})_i$ satisfies  the KKT conditions.  %, we have $\nu_i\geq0$, $p_i\geq0$ and $\nu_i p_i=0$ (complementary slackness). 
	
	Next, we address the existence $\theta$ in $(-1/a_{\max},\infty)$ such that $(p_{\theta i})$ satisfies the linear constraint.
	Let 
	\(C_\theta =\sum_{i=1}^S a_i p_{\theta i}\).
	Then,
	\[
	\nabla_\theta \log C_\theta =
		\frac{S_1S_4-S_2S_3}{S_2S_4}, 
	\]
	where 
	\(S_1=\sum_i a_i ^2(1+(q-1)\theta a_i)^{\frac{2-q}{q-1}}, S_2=\sum_i a_i(1+(q-1)\theta a_i)^\frac{1}{q-1},
	S_3=\sum_i   a_i(1+(q-1)\theta a_i)^\frac{2-q}{q-1} , S_4=\sum_i  (1+(q-1)\theta a_i)^\frac{1}{q-1}.
	\)
	Hence,
	\[
	{S_1S_4-S_2S_3}=\sum_{i,j}(1+(q-1)\theta a_i)^{\frac{2-q}{q-1}} (1+(q-1)\theta a_j)^\frac{2-q}{q-1}
	(a_i^2-a_ia_j)
	\]
	If $\tilde p_i =(1+(q-1)\theta a_i)^{\frac{2-q}{q-1}}/\sum_{j}(1+(q-1)\theta a_j)^{\frac{2-q}{q-1}}$, then
	\[
	{S_1S_4-S_2S_3} \propto \sum_i a_i^2 \tilde p_i-\Big(\sum_i a_i \tilde p_i\Big)^2.
	\]
	This implies ${S_1S_4-S_2S_3}\geq0$, or equivalently $\nabla_\theta \log C_\theta\geq0$, and hence
	$C_\theta$ is monotone increasing in $\theta$.
	Therefore, $C_{\bms a}\leq C_\theta\leq a_{\max}$ noting $\lim_{\theta\rightarrow
-\infty}C_\theta =C_{\bms a}$.
	This ensures that $(q-1)\theta$ is in $(-1/a_{\max},\infty)$ such that $(\hat p_{\theta i})$ satisfies the linear constraint. 
	
	\noindent
	Case 2:  $q> 1$.
	Similarly, by the monotonicity argument, the Lagrangian is the same as $\mathcal L$, and we can observe the similar argument for the KKT conditions.
	If $p_i>0$, then $\nu_i=0$, and thus $ p_i^{q-1}=\mu+\lambda a_i$.
	If $p_i=0$,  then the stationarity condition yields $\nu_i=-(\mu+\lambda a_i) \geq0$.
	This yields the equilibrium distribution is given by $\hat p_{\theta i}$ in \eqref{hat-p2},  %\eqref{equili}.
	which satisfies the KKT conditions.
	Since $C_\theta$ is monotone increasing, $C_\theta$ coverges to $C_{\bms a}$ as $theta$ goes to $\infty$.
	Therefore, $a_{\min}\leq C(\theta)\leq C_{\bms a}$, which ensures  the existence $\theta$ in $(-1/a_{\min},\infty)$ such that $(\hat p_{\theta i})$ satisfies the linear constraint. 
	The proof is complete.
\end{proof}
We note that, if $q=1$, then the equilibrium is well-known as the maximum entropy distribution:
\[
\hat p_{\theta i} = \frac{ \exp(\theta a_i)}{\sum_{j=1}^S \exp(\theta a_j)} \quad (i=1,...,S),
\]
in which the equilibrium distribution can be defined for any $C$ of the fully feasible interval $ (a_{\min}, a_{\max})$.  
This is referred to as  a softmax function that is widely utilized in the field of Machine Learning.
It also closely related to MaxEnt model that is important as a species distribution model,
see \cite{phillips2004}.

To illustrate the application of Proposition~\ref{prop2}, we consider a community with $101$ species, where the trait $a_i$ for species $i$ is given by 
$$
a_i=3\exp\{-(i-20)^2/100\}+\exp\{-(i-80)^2/100\} \quad (i=1,...,101).
$$
This trait distribution creates two peaks, representing species with higher ecological advantages. 
We compute the maximum ${}^q\!D$ distributions under this linear constraint $\sum_{i=1}^{101}a_i p_i=C$ {with $C=0.2$ and $C=0.5$ for $q=2$ and $q=0.7$, respectively.} 
The results  are shown in Fig. \ref{Max-hill-D}.  
The left panel of Fig. \ref{Max-hill-D}displays the trait values $a_i$ against species $i$.
The right panel shows  the optimized species distribution for $q=2$ (squares) and $q=0.7$ (circles).     
For $q=2$, the distribution  assigns zero probability to species with indices between 
$i=11$ and $i=30$, indicating that these species are excluded from the optimal distribution due to the emphasis on  common species.
For $q=0.7$,  the distribution includes all species with non-zero probabilities, promoting evenness across the species. 
Indeed, for $q=0.7$, the Hill number is more sensitive to rare species, resulting in a more uniform distribution of probabilities among the species.
We note that this numerical study can be conducted by solving only a one-parameter equation for $\theta$ to satisfy the linear constraint thanks to the result of Proposition \ref{prop2}.  
Thus, the optimization problem with 101 variables is drastically reduced to such a simplified one.

\begin{figure}[htbp]
	\centering
	\includegraphics[%bb=-100 0 1000 1000,
width=130mm]{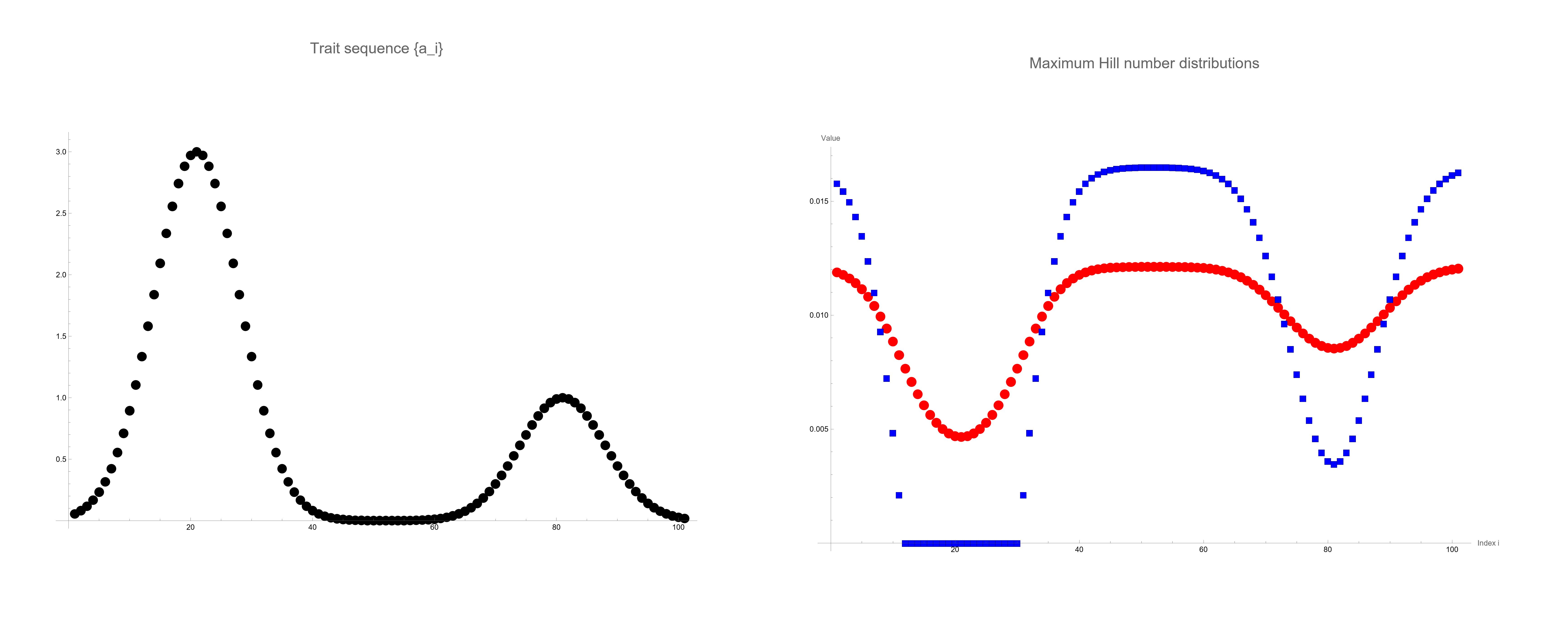}
%	\vspace{-2mm}
	\caption{The Max Hill number distributions}
	\label{Max-hill-D}
\end{figure}
We discuss ecological meaning for $C_{\bms a}$, which represents a critical value in the linear constraint.  
% The limit reflects an aggregate measure of species traits influencing the community. 
If $C$ satisfy $C<C_{\bms a}$ for a case of $q<1$, then the optimal distribution becomes  degenerate,  meaning that it assigns zero probability to some species.   
This occurs because the resource constraint  $C$ is too restrictive to allow for a distribution where all species have positive abundances. 
Ecologically, this reflects situations where only species with certain trait values $a_i$ can survive under the given constraints, leading to competitive exclusion and reduced diversity.
Thus, the ecosystem cannot sustain all species at non-zero abundances without violating resource constraints.
Fig. \ref{Max-hillC} gives a 3-D plot of the Hill number ${}^q\!D({\bm p})$ against $\bm p$ of a $2$-dimensional simplex with $q=1$. In the surface, the curve $\{\hat{\bm p}_\theta\}_\theta$ has the centered point that attains the global maximum $3$ of ${}^q\!D({\bm p})$.
\begin{figure}[htbp]
	\begin{center}
		%\hspace{28mm}
		\includegraphics[%bb=0 0 500 400,
width=100mm]{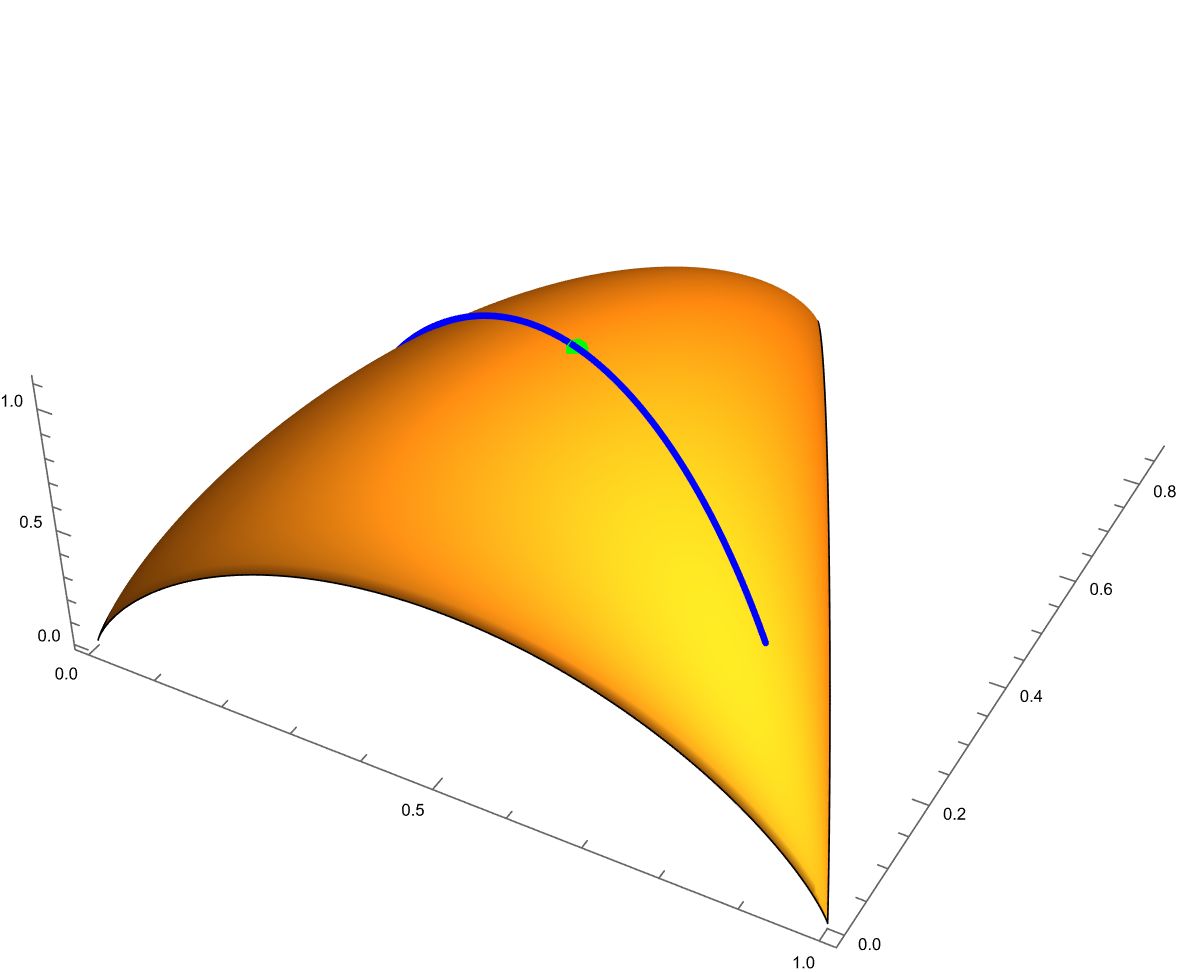}
	\end{center}
	\vspace{-2mm}\caption{Plots of the Max Hill number}
	\label{Max-hillC}
\end{figure}
Fig. \ref{Max-hillB} gives a contour plot of the Hill number ${}^q\!D({\bm p})$ against $p$ of a $2$-dimensional simplex with $q=1$. In the simplex, the curve represents the model $\{{\bm p}_\theta\}_\theta$ intersects the line defined by the linear constrain at the point with the conditional maximum; the centered point attains the global maximum $3$ of ${}^q\!D({\bm p})$.
\begin{figure}[htbp]
	\begin{center}
		%\hspace{28mm}
		\includegraphics[%bb=0 0 500 400, 
width=100mm]{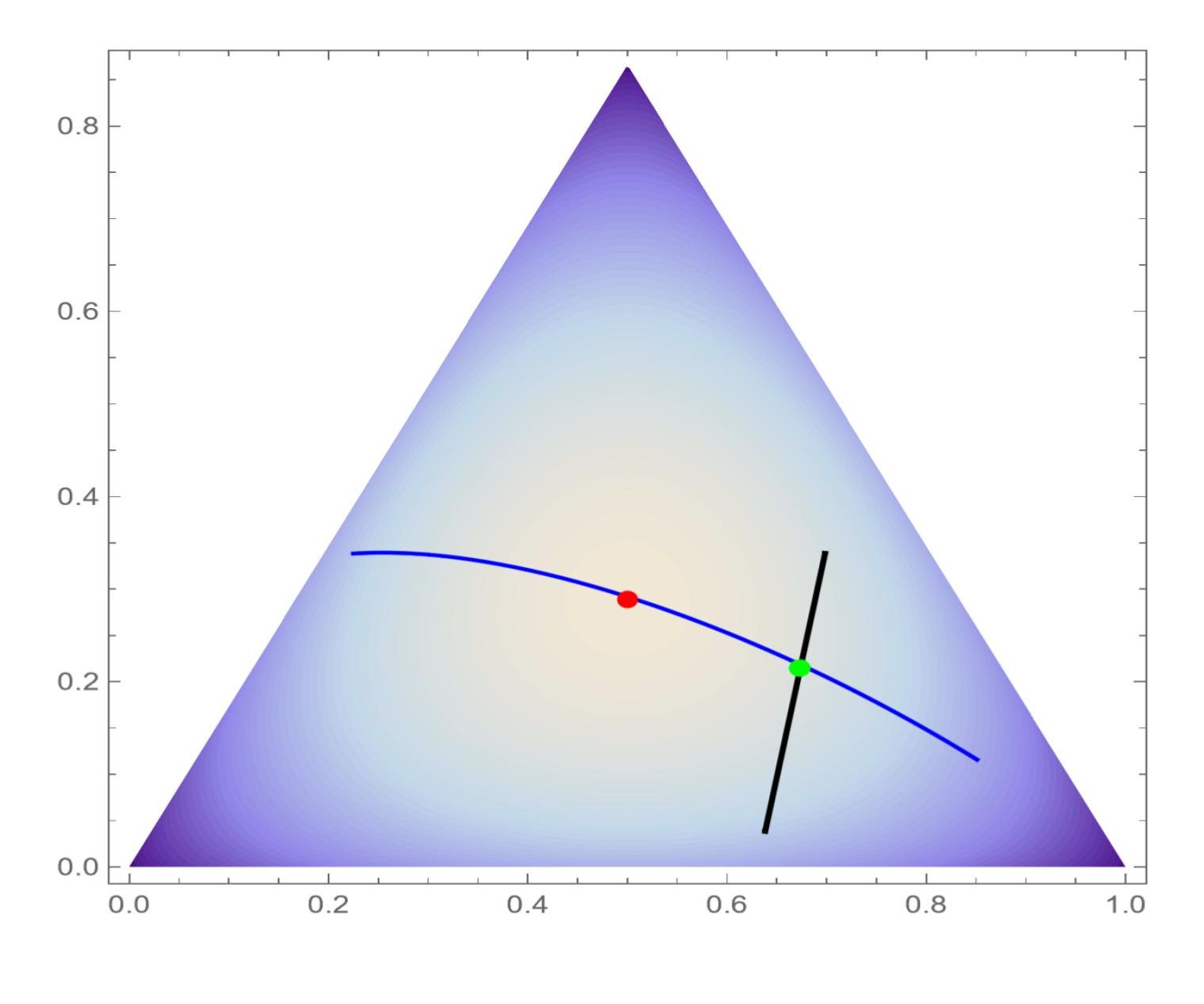}
	\end{center}
	\vspace{-2mm}\caption{Plots of the Max Hill number}
	\label{Max-hillB}
\end{figure}

On the other hand, Rao's quadratic entropy stands out for its ability to incorporate not only species abundance but also functional and genetic dissimilarities.
Let us look into the distribution maximizing Rao's quadratic entropy.

\

\begin{proposition}\label{prop3}
	Assume that $\bm W$ is invertible and that $\bm W^{-1}{\bf 1}>{\bf0}$ in Rao's quadratic entropy ${\rm Q}({\bm p})={\bm p}^\top {\bm W} {\bm p}$, where ${\mathbf 1}$ denotes a copy vector  of $1$'s.  
	Then, ${\rm Q}({\bm p})$ is maximized at
	\begin{align}\label{p-Q}
		\hat {\bm p}_0^{({\rm Q})}=\frac{{ \bm W}^{-1}{\mathbf 1} }{{\mathbf 1}^\top  {\bm W}^{-1}{\mathbf 1} }.
	\end{align}
\end{proposition}
\begin{proof}
	We introduce the Lagrangian function with the multipliers \( \mu \) and $\bm\nu$ for $\bm p$ to be in 
	$\Delta_{S-1}$:
	\[
	\mathcal{L} = \frac{1}{2}{\bm p}^\top \bm W{\bm p} - \mu \left( {\mathbf 1}^\top {\bm p} - 1 \right)-\bm\nu^\top \bm{p}.
	\]
	The  stationarity condition yields ${\bm W}{\bm p}-\mu {\bf 1}+\bm\nu={\bf0}$.
	If $p_i=0$, then $\mu_i=0$, and we have $({\bm W}\bm p)_i=\mu.$
	Therefore, $\bm W\bm p=\mu{\bf1}.$
	Solving for $\bm p$, we get 
	$
	\bm p={\bm W^{-1}{\bf1}}/({{\bf1}^\top \bm W^{-1}{\bf1}}).
	$
	The positivity of $\bm W^{-1}{\bf1}$ ensures that $\bm p\geq0$ and the KKT conditions are satisfied.
\end{proof}
We discuss a case where entries of $\bm W^{-1}{\bf 1}$ have positive and negative signs.
The optimal distribution would be in the boundary of $\Delta_{S-1}$, however any closed form is not solved here.
For example, consider the following examples of a dissimilarity matrix:
\[
\bm W_1=\begin{bmatrix}
	0&1& 2& 3\\
	1& 0 &3& 2\\
	2& 3& 0& 3\\
	3& 2 &2& 0
\end{bmatrix},
\quad
\bm W_2=\begin{bmatrix}
	0&1& 2& 3\\
	1& 0 &2& 1\\
	2& 2& 0& 1\\
	3& 1 &1& 0
\end{bmatrix}.
\]
Then, we find  $\bm W_1^{-1}{\bf 1}=(\frac{3}{23},\frac{3}{23},\frac{4}{23},\frac{4}{23})^\top$, which satisfies the assumption Proposition \ref{prop3}.   The maximum of Rao's entropy defined by $\bm W_1$  is $\frac{14}{23}$ at the optimal distribution $(\frac{3}{14},\frac{3}{14},\frac{2}{7},\frac{2}{7})^\top$.
On the other hand,  $\bm W_2^{-1}{\bf 1}=(-2,4,3,-3)^\top$, which violates the assumption Proposition \ref{prop3}. 
Indeed, the optimal distribution is given by $(\frac{1}{2},0,0,\frac{1}{2})^\top$ in the boundary, at which Roa's entropy has a maximum $\frac{3}{2}$.
These species with zero probabilities are effectively excluded from the ecosystem's diversity under the entropy maximization principle, meaning they "cannot survive" within the model's framework.

If all species are equally dissimilar, then the global maximum distribution $\hat{p}_0^{(\rm Q)}$ equals the even distribution $\bm e$. 
In effect, consider the case where ${\bm W}= {\bf1}{\bf1}^\top - {\bf I}$, where $\bf I$ denotes the identity matrix.
Then, $
{\bm W}^{-1}=\frac{1}{S-1}{\bf1}{\bf1}^\top-{\bf I},
$ 
and hence, $\hat{\bm p}_0^{\rm (Q)}$ equals the even distribution ${\bm e}$.
This is the same as the maximizer of ${}^qD({\bm p})$ noting that Rao's quadratic entropy is reduced to the Gini-Simpson index.
However, in most realistic scenarios, species differ in their functional or genetic traits, leading to variations in dissimilarities. Therefore, $\hat{\bm p}_0^{(\rm Q)}$ typically differs from $\bm e$, giving higher probabilities to species more dissimilar to others.
Next, consider  another simple case where
\(
{\bm W}_{\bms\pi}= {\bm \pi}{\bm\pi}^\top - {\rm diag}({\bm \pi}^{2})
\)
for a fixed $\bm\pi$ of $\Delta_{S-1}$, where $\rm diag$ denotes the diagonal matrix.
Thus, 
$$
{\bm W}_{\bms\pi}^{-1}=\frac{1}{S-1}{\bm\pi}^{-1}({\bm\pi}^{-1})^\top -{\rm diag}({\bm \pi}^{-2}),
$$ 
and hence $({\bm W}_{\bms\pi}^{-1}{\bf1})_i=\frac{{\bf 1}^\top \bms{\pi}^{-1}}{S-1}{\pi_i}^{-1} -{ \pi_i}^{-2}$.
This yields that, if $\pi_i > \frac{S-1}{{\bf 1}^\top \bms{\pi}^{-1}}$, then $({\bm W}_{\bms\pi}^{-1}{\bf1})_i <0$, or the probability of species $i$ is a zero in the maximum distribution.
Further, assume $\bm\pi=(\alpha \tilde{\bf1},\beta \tilde{\bf1})/m(\alpha+\beta)$, where $S=2m$ and $\tilde{\bf1}$ is the one-copy vector of $m$ dimension.   Then, a straightforward calculus yields that, if 
$$
\Big|\frac{\alpha}{\alpha+\beta}-\frac{1}{2}\Big|\leq \frac{1}{2(2m-1)},
$$
then $\hat{\bm p}_0^{(\rm Q)}$ defined in \eqref{p-Q} is in $\Delta_{S-1}$.
This situation is quite near the case of $\alpha=\beta$, or the equally-dissimilar case $\bm W$.
In general, it is difficult to get a closed form of the maximum distribution without the assumption ${\bm W}^{-1}{\bf1}>{\bf0}$. 
However, we can use efficient algorithms to find numerically the maximizer, in which the optimization problem is reduced a a standard quadratic programming problem.  
Fast solvers like interior-point methods, active set methods, or alternating direction method of multipliers can handle this formulation. 
Tools such as quadprog in R and Python or specialized libraries like Gurobi or CVXPY provide robust implementations for high-dimensional problems.
See \cite{dostal2009optimal} for quadratic programming algorithms. 
%Dostál, Z. (2009). Optimal quadratic programming algorithms: with applications to variational inequalities (Vol. 23). Springer Science & Business Media.

Next, consider a linear constraint: ${\bm a}^\top {\bm p}=C$, where ${\bm a}\geq\bf{0}$ and $C$ are fixed constants.
Then,  under the linear constraint,  %if  $\bm{W}^{-1}\bm{a}\geq\bf{0}$,
Similarly, we introduce the Lagrangian function:
\[
\frac{1}{2}{\rm Q}({\bm p})+  \mu \left( {\mathbf 1}^\top {\bm p} - 1 \right) +\lambda\left( {\bm a}^\top {\bm p} - C \right).
\]
The  stationarity equation is given by
\(
{\bm W}{\bm p} + \lambda {\bm a} + \mu {\mathbf 1} = {\bf 0},
\)
which yields the solution form:
\begin{align}\label{maximizerR}
	\hat {\bm p}_\theta^{({\rm Q})}= \frac{{\bm W}^{-1}({\bf 1}+\theta{\bm a})}{{\bf1}^\top({\bm W}^{-1}({\bf 1}+\theta{\bm a}))}.
\end{align}
Here $\theta$ is determined by the linear constraint, so that
\[
\theta =\frac{(C{\bf1}-\bm{a})^\top{\bm W}^{-1}{\bf1}}{(C{\bf1}-\bm{a})^\top{\bm W}^{-1}{\bm a}}.
\]
If we assume ${\bm W}^{-1}({\bf 1}+\theta{\bm a})>0$, then  $\hat {\bm p}_\theta^{({\rm Q})}$ is properly the distribution maximizing Rao's entropy due to the discussion similar to the proof of Proposition \ref{prop3}.

Let us examine the maximum Leinster-Cobbold index distribution.
This index is integrated the Hill number with Rao's quadratic entropy, providing a unified framework for measuring biodiversity that accounts for both species abundances and similarities..
We can find the maximizing distribution by combining arguments in Propositions \ref{prop2} and \ref{prop3}.

\begin{proposition}\label{prop4}
	Assume that $\bm Z$ is invertible and that $\bm{Z}^{-1}\bf{1}>\bf{0}$ in  the Leinster-Cobbold index:
	\[
	{}^q\!D_{\bms Z}({\bm p}) = \left( {\bm p}^\top ({\bm Z}{\bm p})^{q-1} \right)^{\frac{1}{1 - q}},
	\]
	where  ${\bm y}^x$ is defined as the vector $(y_i^x)$ for a vector ${\bm y}=(y_i)$.
	Then, the  distribution maximizing ${}^q\!D_{\bms Z}({\bm p})$ in $\bm p$ is given by 
	\begin{align}\label{maxLC}
		\hat{\bm p}_{\rm LC}^{(q)}=\bm Z^{-1}{\bf1}/({\bf1}^\top\bm Z{\bf1}).
	\end{align}
	
\end{proposition}
\begin{proof}
	The Lagrangian function with the multipliers \( \mu \) and  \( \nu \) is given by
	\[
	{\bm p}^\top ({\bm Z}{\bm p})^{q-1}- \mu \left( {\mathbf 1}^\top {\bm p} - 1 \right) +
	{\bm \nu}^\top {\bm p},
	\]
	which leads to the  stationarity equation, 
	\[
	({\bm {Z p}})^{q-1}+(q-1) {\bm Z}(\bm{p}\odot \{({\bm {Z p}})^{q-2}\}-\mu{\bf 1}+{\bm\nu}={\bf0},
	\]
	where $\odot$ denotes the Hadamar product.
	If $\bm{p}>{\bf0}$, then $\bm \nu={\bf0}$ by the complementary slackness.
	Let $\bm{p}=\sigma\bm{Z}^{-1}{\bf 1}$, where $\sigma>0$.  Then, the  stationarity  equation becomes
	\[
	\sigma^{q-1}{\bf1}+(q-1) \sigma^{q-1} {\bf1}-\mu{\bf 1}={\bf0},
	\]
	which implies $\bm{p}\propto\bm{Z}^{-1}{\bf 1}$ is a candidate of the solution.
	Therefore, the  maximum distribution is given by \eqref{maxLC}.
	
\end{proof}
It is worthwhile to note that $\hat{\bm p}_{\rm LC}^{(q)}$ is independent of $q$, and is the same form as the maximum Rao's entropy distribution $\hat{\bm p}_0^{\rm (Q)}$ defined in \eqref{p-Q}.
The distribution maximizing  under a linear constraint ${\bm a}^\top {\bm p}=C$ 
is given by
\begin{align}\label{maximizer2}
	\hat {\bm p}_\theta^{({\rm LC})}=\frac{{\bm Z}^{-1}\big({\bf1}+(q-1)\theta{\bm a}\big)^\frac{1}{q-1}}{{\bf1}^\top{\bm Z}^{-1}\big({\bf1}+(q-1)\theta{\bm a}\big)^\frac{1}{q-1}},
\end{align}
where $\theta$ is determined by $C$.
We confirm that, if $q=2$, then $\hat {\bm p}_\theta^{({\rm LC})}$ is reduced to $\hat {\bm p}_\theta^{({\rm Q})}$ defined in \eqref{maximizerR} by a re-parameterized $\theta$; if ${\bm Z}={\bf I}$, then this is reduced to $\hat{\bm p}_\theta$ in \eqref{hat-p}.
Thus, this effectively integrates both maximum distributions.
Then, the stationarity condition is written by
\[
q(\bm {Z p})^{q-1} + \lambda {\bm a} + \mu {\mathbf 1} = 0.
\]
This gives $(\bm {Z p})^{q-1} \propto  {\mathbf 1}+(q-1)\theta{\bm a}$, which concludes the solution form \eqref{maximizer2}. 

The Leinster-Cobbold index serves as a unifying framework that encompasses both the Hill numbers and Rao's quadratic entropy.
The inverse similarity matrix ${\bm Z}^{-1}$
adjusts species abundances based on their functional or phylogenetic relationships.
High $q$ values focus on dominant species or those with higher contributions to similarity; low $q$ values emphasize rare species or those with unique traits.

The maximum diversity distribution derived under the linear constraint \({\bm a}^\top {\bm p} =C\) has ecological interpretations that bridge mathematical optimization and real-world ecological dynamics.
The coefficients \(a_i\) represent species-specific traits or environmental tolerances that directly influence each species' ability to thrive under certain ecological conditions. For instance, in ecosystems where a specific resource is limited, \(a_i\) could represent each species' efficiency in utilizing that resource. Species with higher \(a_i\) values are more efficient and thus have a competitive advantage, leading to higher abundances \(p_i\).
%In harsh environments like deserts or polluted habitats, \(a_i\) might quantify drought or pollution tolerance. Species with higher \(a_i\) can survive and reproduce more effectively under these stressors, resulting in higher representation in the community.
Maximizing the Leinster-Cobbold index under this constraint does not lead to perfect evenness, which is rarely observed in natural ecosystems. Instead, it yields a distribution that balances diversity with ecological realism, reflecting the natural dominance of certain species due to their advantageous traits. This mathematical optimization highlights the trade-offs between achieving maximum diversity and adhering to ecological constraints, demonstrating that the most diverse community under a given constraint proportionally represents species according to their ecological roles and advantages.
This approach offers a quantitative framework to model and analyze community structures under various ecological scenarios. By adjusting \(a_i\) and \(c\), ecologists can simulate different environmental conditions and assess their impact on diversity and species distributions.

%%%%%%%%%%%%%%%%%%%%%%%%%%%%%%%%%%%%%%%%%%%%%%%%%%%%%%%%%%%%%%%%%%%%%%%%%%%%%%%%%%%%%%%%%%%%%%%%%%%%%%%%%%%%%%%%%%%%%%%%%%%%%%%%%%%%%%%%%%%%%%%%%%%%%%%%%%%%%%%%%%%%%%%%%%%%%%%%%%%%%%%%%%%%%%%%%%%%%%%%%%%%%%%%%%%%%%%%%%%%%%%%%%%%%%%%%%%%%%%%%%%%%%%%%

\section{Information geometry on a simplex}\label{sect4}

In this section, we explore the general properties of a simplex within the framework of information geometry. 
Information geometry provides a powerful way to study statistical models by viewing them as Riemannian manifolds endowed with the Fisher-Rao metric \cite{rao1945}. 
The Fisher-Rao metric provides the Cramér-Rao lower bound for unbiased estimators.
This perspective allows us to examine geometric structures such as geodesics, divergence measures, and metric tensors, which have profound implications in statistical inference and diversity measurement.

We have a concise review of the information geometric natures on a simplex.
Let $ \Delta_{S-1}$ be a simplex of dimension $S-1$, or
$$
{\Delta_{S-1}}=\big\{{\bm p} \in \mathbb R^S :{\bf1}^\top \bm{ p}=1,\bm p\geq {\bf0}\big\}.
$$
Henceforth, we identify $\Delta_{S-1}$ with the space of all $S$-variate categorical distributions.
The Fisher-Rao metric \(g \) on \( \Delta_{S-1} \) is given by the information matrix:
\begin{align}\label{FR}
	\bm G_{\bms p} = {\rm diag}({\bm p}^{-1}) + \frac{1}{p_S}{\bf1}{\bf1}^\top
\end{align}
where \( p_S = 1 - \sum_{i=1}^{S-1} p_i \).
The Fisher-Rao metric captures the intrinsic geometry of the probability simplex, reflecting the sensitivity of probability distributions to parameter changes.
The metric $g$ induces geodesics that can be expressed using trigonometric functions, specifically the arcsine function.
The Fisher-Rao distance between two points \( \bm \pi \) and \(\bm p \) on the simplex $\Delta_{S-1}$ is given by:
\[
d({\bm \pi}, {\bm p}) = 2 \arccos\left( \sqrt{\bm \pi}^\top\sqrt{\bm p} \right).
\]
The Riemannian geodesic connecting \( {\bm \pi} \) and \( {\bm p} \) under the Fisher-Rao metric can be expressed parametrically using trigonometric functions:
\[
\bm p (t) = \left( \frac{\sin\left( (1 -t) \phi \right)}{\sin \phi} \sqrt{\bm \pi} + \frac{\sin\left( t \phi \right)}{\sin \phi} \sqrt{\bm p} \right)^2
\]
for $t, 0\leq t\leq 1$, where \( \phi \) represents the angle between \( \sqrt{\bm\pi } \) and \( \sqrt{\bm p} \) on the sphere.
This geodesic provides the shortest path between two probability distributions under the Fisher-Rao metric and is characterized by trigonometric functions arising from the geometry of the unit sphere.
% References:
%- Amari, S.-I., & Nagaoka, H. (2000). *Methods of Information Geometry*. American Mathematical Society.
%- Skovgaard, L. T. (1984). A Riemannian geometry of the multivariate normal model. *Scandinavian Journal of Statistics*, 11(4), 211-223.

In information geometry, two types of affine connections--the mixture and exponential connections--lead to dualistic interpretations between statistical models and inference, see \cite{amari1982,amari-nagaoka2000}, \cite{nielsen2021geodesic}. 
These connections give rise to m-geodesics and e-geodesics:
For any distinct points $\bm \pi$ and $\bm p$ of $\Delta_{S-1}$, the m-geodesic connecting $\bm \pi$ and $\bm p$ is given by
\[
\bm p^{\rm (m)}(t)=(1-t)\bm\pi+t \bm p, \quad t\in[0,1]
\]
the e-geodesic connecting $\bm\pi$ and $\bm p$ is given by
\[
\bm p^{\rm (e)}(t)=\frac{\exp((1-t)\log \bm\pi+t \log \bm  p)}{{\bf1}^\top \exp((1-t)\log \bm\pi +t\log \bm p)}, \quad t\in[0,1],
\]
where $\log \bm p$ denotes the vector of logarithms of $p_i$.
The e-geodesic corresponds to linear interpolation in the natural parameter space of the exponential family.
It is noted that the range of $t$ for  $\bm p_t^{\rm (m)}$  can be extended to a closed interval including $[0,1]$; 
the range of $t$ for $\bm p_t^{\rm (e)}$  can be extended to as $(-\infty,\infty)$, see appendix for detailed discussion for the enlarged interval.
We observe an important example of the mixture geodesic in Rao's quadratic entropy.
In effect, the maximum quadratic entropy distribution in \eqref{maximizerR} under the linear constraint is written as
\begin{align}\nonumber %label{maximizer}
	\hat {\bm p}_t^{({\rm Q})}=(1-t)\hat {\bm p}_0^{({\rm Q})}+t {\bm p}_{\bms a}.
\end{align}
This is a mixture geodesic connecting ${\bm p}_0^{({\rm Q})}$ and ${\bm p}_{\bms a}$, where ${\bm p}_{\bms a}={{\bm W}^{-1}{\bm a}}/({{\mathbf 1}^\top {\bm W}^{-1}{\bm a}})$.
Thus, $\hat{\bm p}_t^{({\rm Q})}$ becomes a global maximum at $t=0$; it becomes a constrained maximum
at $t$ determined by the linear constraint. 
If we take $K$ distinct points ${\bm p}_{(1)},...,{\bm p}_{(K)}$, then the e-geodesic family is given by
\[
\bm p ^{\rm (e)}_{\bms t}=Z_{\bms t}^{\rm (e)}{\exp\big(t_1\log\bm  p_{(1)}\cdots+t_K\log \bm p_{(K)}\big)}
\]
where ${\bm t}=(t_i)$ is in a simplex $\Delta_{K-1}$.
In general, such a family ${\mathcal F}^{\rm (e)}=\{{\bm p}^{(\rm e)}_{\bms t}: {\bm t}\in \Delta_{K-1}\}$
is referred to as an exponential family.
By definition, the e-geodesic connecting any points $\bm\pi$ and $\bm p$ of ${\mathcal F}^{\rm (e)}$ is included in ${\mathcal F}^{\rm (e)}$.
%Rao developed the Cramer-Rao inequality, establishing a bound for the variance of any unbiased estimator.
Under an assumption of an exponential family, the Cramer-Rao inequality provides insight into the efficiency of estimators, making it clear when estimators are at or near this lower bound.
It is shown that MLEs within the exponential family have optimal properties, often achieving the Cramer-Rao lower bound under regularity conditions, see \cite{rao1961}.
Further, in a curved model in ${\mathcal F}_{\exp}$, \cite{rao1961} gives an insightful proof for the maximum likelihood estimator to be second-order efficient, see also \cite{{efron1975defining},{eguchi1983}}.
% Exponential family Rao UNIQUE

As a generalized geodesic, we consider the $q$-geodesic connecting ${\bm \pi}$ with ${\bm p}$ defined by
\[
\bm p^{(q)}_t =Z_t^{(q)}{\{(1-t)\bm \pi^{q-1}+ t \bm p^{q-1}\}^\frac{1}{q-1}}, \quad t\in[0,1],
\]
where $q$ is a fixed real number and $Z_t^{(q)}$ is the normalizing constant, cf. \cite{eguchi2011projective}.
We will observe a close relation to the maximum Hill-number distributions discussed in Section \ref{sect3}.
By definition, the family of the $q$-geodesics includes the mixture geodesic and the exponential geodesic as
\[
\lim_{q\rightarrow0} \bm p^{(q)}_t= \bm p^{\rm (m)}_t,\quad
\lim_{q\rightarrow1} \bm p^{(q)}_t= \bm p ^{\rm (e)}_t.
\]
We note that the feasible range for $t$ can be extended from $[0,1]$ to a closed interval, see Appendix for the exact form.

Similarly, for $K $ distinct points ${\bm p}_{(1)},...,{\bm p}_{(K)}$, then the $q$-geodesic family is given by ${\mathcal F}^{ (q)}=\{\bm p^{\rm (q)}_{\bms t} :  \bm t\in \Delta_{K-1}\}$, where 
\[
\bm p^{\rm (q)}_{\bms t} =Z_{\bms t}^{(q)} \big\{t_1 \bm p_{(1)}^{\ q-1}+\cdots+t_K \bm p_{(K)}^{\ q-1}\big\}^\frac{1}{q-1} . 
\]
The $q$-geodesic connecting between any points $\bm \pi$ and $\bm p$ of ${\mathcal F}^{\rm (q)}$ is included in the $q$-geodesic family.

% \begin{proposition}\label{prop5}
	Consider a linear constraint: $\bm a^\top \bm p=C$ in $\Delta_{S-1}$. 
	Then, the maximum Hill-number distribution
	\[
	\hat{\bm p}_\theta=\frac{({\bf 1}+(q-1)\theta \bm a)^\frac{1}{q-1}}{{\bf 1}^\top ({\bf 1}+(q-1)\theta \bm a)^\frac{1}{q-1}}
	\] 
	with $q<1$ under the linear constraint lies along  the $q$-geodesic connecting between $\bm e$ and $\bm p_{\bms a}$, where  $\bm e={\bf 1}/S$ and 
	\(
	\bm p_{\bms a} =  {\bm a^\frac{1}{q-1}}/{{\bf 1}^\top \bm a^\frac{1}{q-1}}.
	\) 
	% \end{proposition}
Indeed, by the definition of the $q$-geodesic,
\[
\bm p^{(q)}_t =\frac{\big((1-t)\bm e+t \bm p_{\bms a}^{q-1}\big)^\frac{1}{q-1}}
{{\bf 1}^\top\big((1-t)\bm e+t \bm p_{\bms a}^{q-1}\big)^\frac{1}{q-1}} 
\]
which can be written as
\[
\bm p^{(q)}_t =\frac{\big((1-t)T{\bf1}+t \bm a\big)^\frac{1}{q-1}}
{{\bf 1}^\top\big((1-t)T{\bf1}+t \bm a\big)^\frac{1}{q-1}}. 
\]
where $T={({\bf 1}^\top \bm a^\frac{1}{q-1})^{q-1}}/{S^{q-1}}$.
Therefore, $\bm p^{(q)}_t= \hat{\bm p}_\theta$ if $\theta$ is defined by
\(  t/\big((1-t)T(q-1)\big). \)
%\end{proof}
The $q$-geodesic provides a path of distributions that transition smoothly from maximum evenness 
$\bm e$ to a distribution that fully incorporates the constraint $\bm p_{\bms a}$.

Consider a one-parameter family of probability distributions defined as:
\begin{align}\label{max-hill-model}
	{\mathcal P}_{\bms a}^{(q)}=\{\bm p_\theta^{(q)}=Z_{\theta}^{(q)}{({\bf1}+(q-1)\theta \bm a )^\frac{1}{q-1}}  : \theta\in\Theta\}
\end{align}
and a subset of of $\Delta_{S-1}$ constrained by a linear condition:
\[
{\mathcal H}_{\bms a}(  C)=\big\{{\bm p}\in\Delta_{S-1}:\bm a^\top \bm p = C\big\}.
\]
%assuming $\bm a \geq {\bf0}$, and $C=\bm a^\top \bm{\hat\pi}$, where $\bm{\hat\pi}$ is the observed frequency.
Geometrically, this implies that the simplex $\Delta_{S-1}$ can be foliated into a family of hyperplanes 
$\{{\mathcal H}_{\bms a}(C)\}$, where each hyperplane intersects the one-parameter family ${\mathcal P}_{\bms a}^{(q)}$ at a single point:
$$
{\mathcal H}_{\bms a}(C) \cap {\mathcal P}_{\bms a}^{(q)} = \{\bm p_{\theta^*}^{(q)}\},
$$
where $\theta^*$ satisfies the condition $\bm a^\top \bm p_{\theta^*}^{(q)} = C$.
This foliation provides a geometric perspective on the simplex: each hyperplane ${\mathcal H}_{\bms a}(C)$ forms a "leaf," and the distribution of these leaves varies smoothly over $\Delta_{S-1}$. 
At every point within a leaf, the tangent space of the leaf aligns with a smoothly varying distribution over the simplex.
Fig. \ref{Max-hill-A} gives a 3D plot of the foliation in a $3$-dimensional simplex $\Delta_3$.
%The foliation is observed Hill number ${}^q\!D({\bm p})$ against $p$ of a $2$-dimensional simplex with $q=1$.
We observe that the curve represents the ray ${\mathcal P}_{\bms a}$ that intersects the  leaves ${\mathcal H}_{\bms a}(C)$'s  at the conditional maximum points, in which the leaves are parallel to each other. %; the centered point attains the global maximum $4$ of ${}^q\!D({\bm p})$.
\begin{figure}[htbp]
	\begin{center}
		\hspace{10mm}
		\includegraphics[%bb=0 0 500 500, 
width=100mm]{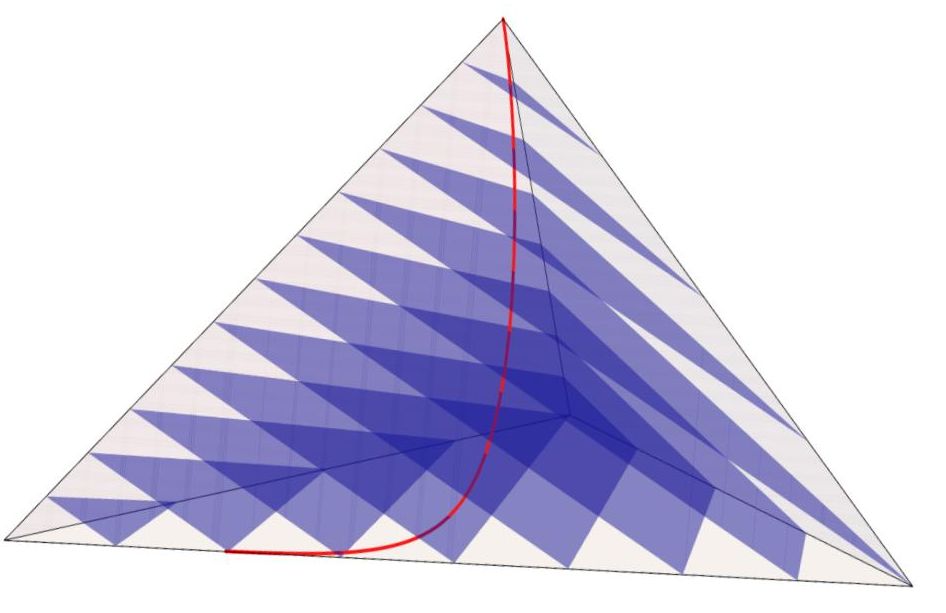}
	\end{center}
	\vspace{-2mm}\caption{Plots of the foliation in $\Delta_3$.}
	\label{Max-hill-A}
\end{figure}

%\%どこでいう？Thus, this is associated with the $K$-dimensional $q$-geodesic model.

The concepts for divergence have been well discussed in statistics, see \cite{{burbea1982convexity}, {basu1998robust}, {eguchi2022minimum}}.
The $\gamma$-power cross-entropy between community distributions $\bm \pi $ and $\bm p$ can be given by
\[
H_\gamma({\bm \pi},{\bm p}) = -\frac{1}{\gamma}\frac{\!\!\!\!\bm\pi^\top \bm{ p}^\gamma}
{\big({\bf1}^\top\bm p^{\gamma+1}\big)^\frac{\gamma}{\gamma+1}}
\]
with a power parameter $\gamma$ of $\mathbb R$ and the (diagonal) entropy is given by
\[
H_\gamma({\bm p}) = -\frac{1}{\gamma}{\big({\bf1}^\top\bm p^{\gamma+1}\big)^\frac{1}{\gamma+1}}
\]
equating $\bm\pi$ with $\bm p$, see \cite{{Fujisawa2008},{eguchi2024}}.
When $\gamma=0$, $H_\gamma({\bm p})$ is the Boltzmann-Shannon entropy; when $\gamma=-1$, $H_\gamma({\bm p})$ is the geometric mean,
$\prod_{i=1}^S p_i^\frac{1}{S}$; when $\gamma=-2$, $H_\gamma({\bm p})$ is the harmonic mean, or
$\big(\sum_{i=1}^S p_i^{-1}\big)^{-1}$.
In essence, the $\gamma$-power entropy is  equivalent to the Hill number with a relation of
${}^q\!D({\bm p})=-\gamma \{H_\gamma({\bm p})\}^{-\frac{1}{\gamma}}$ in relation to $q=\gamma+1$.
A fundamental inequality holds: $H_\gamma({\bm \pi},{\bm p})\geq H_\gamma({\bm p})$, where the difference is referred as the $\gamma$-power divergence
\[
{\mathcal D}_\gamma({\bm \pi},{\bm p})=H_\gamma({\bm \pi},{\bm p})-H_\gamma({\bm \pi}).
\]
In general, any statistical divergence is associated with the Riemannian metric and a pair of affine connections, see \cite{eguchi1992geometry}.
The power cross-entropy has the empirical form for regression and classification model, in which the minimization problem gives a robust estimation in a context of machine learning.
When $\gamma$ equals $0$, then $H_\gamma(\bm p)$, $H_\gamma(\bm p,\bm \pi)$ and $D_\gamma(\bm p,\bm\pi)$ are reduced to the Boltzmann-Shannon entropy, the cross-entropy and Kullback-Leibler divergence, respectively.
The empirical form of the cross-entropy under a parametric model is nothing but the negative log-likelihood function,   
see \cite{eguchi2006interpreting} for detailed discussion.
We explore statistical properties for estimation methods for the maximum Hill-number model
${\mathcal P}_{\bms a}^{(q)}$ defined in \eqref{max-hill-model}.

Let $\hat{\bm\pi}$ be an observed frequency vector in $\Delta_{S-1}$ derived from ecological research or auxiliary information, such as prior field surveys, species inventories, or environmental modeling outputs.
Then, the linear constraint is fixed by $\bm  a^\top \bm p=\hat C$,
where $\hat C=\bm a^\top \hat{\bm \pi}$.
In general, the maximum likelihood method is recognized as a universal method for estimating a general parametric model.
Consider the standard  maximum entropy model 
\begin{align}\nonumber
\bm p_\theta^{(1)} = { \exp(\theta\bm a)}/{{\bf1}^\top \exp(\theta\bm a)}.
\end{align}
Then, the log-likelihood function \(
L(\theta)=\hat{\bm\pi}^\top \log(\bm p_{\theta}^{(1)})
\)
leads to the likelihood equation: \(\bm{a}^\top \bm p_{\theta}^{(1)}=\hat C.\)
In effect, the maximum likelihood estimation is equivalent to the maximum Hill number method.
However, other estimation methods are not associated with such equivalence.
In general, we would like to consider the minimum $\gamma$-power method by fixing the power parameter as
$\gamma=q-1$ for the order $q$ of the Hill number.
For an application to ecological studies, we consider the cross Hill number of $q$-order as
\begin{align}\label{q-risk}
	{}^q\! D(\bm\pi,\bm p)=
	{\big({\bf1}^\top  \bm p^{q}\big)}\big({{\bm\pi}^\top  \bm p^{q-1}}\big)^\frac{q}{1-q}.
\end{align}
If $\bm\pi=\bm p$, then  ${}^q\! D(\bm\pi,\bm p)=({\bf1}^\top\bm \pi^{q})^\frac{1}{1-q}$, or the cross Hill number  is reduced to  the Hill number, or ${}^q\! D(\bm \pi)$.
Hence, since $H_\gamma(\bm\pi,\bm p)\geq H_\gamma(\bm\pi)$,
$ %\begin{align}\nonumber
	{}^q\! D(\bm\pi,\bm p)\geq {}^q\! D(\bm\pi),
$ %\end{align}
and hence we define the Hill divergence
$$ %\begin{align}\nonumber
	^q\! {\bm\Delta}(\bm\pi,\bm p)={}^q\! D(\bm\pi,\bm p) -{}^q\! D(\bm\pi).
$$ %\end{align}
Assume $\bm a^\top \bm p_{\hat\theta}^{(q)}= \hat{C}$ and $\bm a^\top \bm\pi=\hat C$.
Then, we result 
$$
{}^q\! D(\bm p_{\hat\theta}^{(q)})\geq{}^q\! D(\bm\pi)
$$
that is, $\bm p_{\hat\theta}^{(q)}$ is the maximum Hill number distribution.   This is because
\begin{align}\nonumber
{}^q\! D(\bm p_{\hat\theta}^{(q)})
-	{}^q\! D(\bm p_{\hat\theta}^{(q)})={}^q\! {\bm\Delta}(\bm p_{\hat\theta}^{(q)},\bm\pi).
\end{align}
due to $\bm\pi^\top (\bm p_{\hat\theta}^{(q)})^{q-1}=\bm p_{\hat\theta}^{(q)}{}^\top(\bm p_{\hat\theta}^{(q)})^{q-1}=1+\theta \hat C$.
In this way, we can show Proposition \ref{prop2} from ${}^q\! {\bm\Delta}(\bm p_{\hat\theta}^{(q)},\bm\pi)\geq0$.

We propose the  maximum $q$-order  estimator  defined by
\begin{align}\label{gamma}
	\hat\theta^{(q)} = \argmin_{\theta} {}^q\! D(\hat{\bm\pi},\bm p_\theta^{(q)}).
\end{align}
We will observe that the  maximum $q$-order  estimation is equivalent to the method of the maximum Hill number distribution.

\

\begin{proposition}\label{prop5}
	Let $\hat\theta^{(q)}$ be the  maximum $q$-order  estimator defined in \eqref{gamma}.
	Then,  the estimator  $\theta =\hat\theta^{(q)}$ satisfies the linear constraint $\bm a^\top \bm p_{\theta}^{(q)}=\hat C$, where $\hat C=\bm a^\top\bm{\hat\pi}$.
	Furthermore,
	\[
	\max_\theta  {}^q\!D(\hat{\bm\pi},\bm p_\theta^{(q)})=\max_{\bms p\> :\>  \bms{a}^\top\bms p\>=\>\hat{C}}\ {}^q\! D(\bm p)
	\]
	
\end{proposition}
\begin{proof}
	The $q$-order cross Hill number \eqref{q-risk}  is simplified at $(\hat{\bm\pi},\bm p_\theta^{(q)})$, noting the definition of $\hat C$:
	\begin{align}\nonumber
		{}^q\!D(\hat{\bm\pi},\bm p_\theta^{(q)}) =
		(1+(q-1)\theta \hat C )^\frac{q}{1-q}\big({\bf1}^\top ({\bf1}-(q-1)\theta\bm a)^\frac{q}{q-1}\big).
	\end{align}
	Hence, the gradient is given by
	\begin{align}\nonumber
		\nabla_\theta {}^q\! D(\hat{\bm\pi},\bm p_\theta^{(q)})=\frac{q}{1-q}\Big\{
		\hat C (1+\tilde\theta \hat C )^\frac{1}{1-q}\big({\bf1}^\top ({\bf1}-\tilde\theta\bm a)^\frac{q}{q-1}\big)
		-(1+\tilde\theta \hat C )^\frac{q}{1-q}\big({\bm a}^\top ({\bf1}-\tilde\theta\bm a)^\frac{1}{q-1}\big)\Big\}
	\end{align}
	which is written as
	\begin{align}\nonumber
		\frac{q}{1-q}
		{  (1+\tilde\theta \hat C )^\frac{1}{1-q}\big({\bf1}^\top ({\bf1}-\tilde\theta\bm a)^\frac{q}{q-1}\big)}
		\big( \hat C- \bm a^\top \bm p_\theta^{(q)}\big)
	\end{align}
	due to the definition of $\bm p_\theta^{(q)}$, where $\tilde\theta=(q-1)\theta$.
	Equating the gradient to $0$ is equivalent to $\bm a^\top \bm p_{\theta}^{(q)}=\hat C$.
	The linear constraint is exactly satisfied at $\theta =\hat\theta^{(q)}$ since, by definition,  $\hat\theta^{(q)}$ is the solution of the gradient equation.
	Further, if $\theta =\hat\theta^{(q)}$, then $1+\theta\hat C= \bm p_\theta^{(q)}{}^\top({\bf1}-\theta\bm a)$,
	so that
	\begin{align}\nonumber
		{}^q\!D(\hat{\bm\pi},\bm p_{\theta}^{(q)}) =
		{\big({\bf1}^\top (\bm p_{\theta}^{(q)})^{q}\big)}\big({{\bm p_{\theta}^{(q)}}^\top  (\bm p_{\theta}^{(q)})^{q-1}}\big)^\frac{q}{1-q}.
	\end{align}
	This is nothing but Hill number ${}^q\!D(\bm p_{\theta}^{(q)}) $ that attains the maximum at $\theta =\hat\theta^{(q)}$ under the linear constraint.
	This completes the proof.
\end{proof}
In the proof, we observe a surprising property: $\nabla_\theta {}^q\! D(\hat{\bm\pi},\bm p_\theta^{(q)})\propto \hat C- \bm a^\top \bm p_\theta^{(q)}$.
This directly shows the equivalence between the maximum $q$-order estimation and the maximum distribution for the Hill number of the order $q$.
Thus, the maxim $q$-order estimation naturally introduce the maximum Hill number of the order $q$.
%In effect, we have the following inequality 
%\begin{align}\nonumber
%^q\!D(\hat{\bm\pi}) \leq {}^q\!D(\hat{\bm\pi},\bm p_{\theta}^{(q)})\leq {}^q\!D(\hat{\bm e})  
%\end{align}

We consider a more general setting of linear constraint that allows multiple linear constraints. Let ${\bm A}$ be a matrix of size $K\times S$ with full rank and
\begin{align}\label{constraint}
	{\mathcal H}_{\bms A}(\bm C)=\{{\bm p}  : {\bm A} {\bm p} ={\bm C}\}
\end{align}
where $\bm C$ is a constant vector of $K$ dimension.
This means $K$ constraints $\bm a_j\top \bm p=C_j\ (j=1,...,K)$, where $\bm a_j$ is the $j$th row vector of $\bm A$. 
Consider the maximum Hill-number distribution under such  linear constraints.
We have an argument similar to that for one-dimensional constraint as discussed above. %in Section \ref{sec3}:
Thus, we introduce Lagrange multipliers \( \mu \), \( {\bm\lambda}  \) and \( {\bm\nu}  \) to incorporate the constraints, and form the Lagrangian function:
\[
\mathcal{L} = \sum_{i=1}^S p_i{}^q + {\bm\lambda}^\top \left({\bm A}{\bm p} -{\bm C} \right) + \mu \left( {\bf1}^\top \bm p  - 1 \right)+\bm\nu^\top \bm p.
\]
The stationarity condition gives the solution form:
\[
\bm p ^{(q)}_{\bms\theta} = \frac{\big [{\bf1}+(q-1){\bm A}^\top{\bm \theta}\big]_+^{\frac{1}{q - 1}}}{{\bf1}^\top \big[{\bf1}+(q-1){\bm A}^\top{\bm\theta}]_+^\frac{1}{q - 1}}.
\]
combining both cases of $q<1$ and $q>1$.  
It is necessary to define a region of feasible $\bm C$ values to ensure that $\bm\theta$ satisfies the linear constraints. However, we omit this discussion as it closely parallels the proof of Proposition \ref{prop2}.
The gradient of the $q$-order risk function for the model $\{\bm p_{\bms\theta}^{(q)}\}_{\bms\theta}$ satisfies
\[
\nabla_{\bms\theta\ } {}^{q}\!D(\hat{\bm\pi},\bm p ^{(q)}_{\bms\theta})\propto \hat{\bm C} -{\bm A} {\bm p}_{\bms{\theta}}^{(q)}
\] 
Hence, the maximum $q$-order estimator $\hat{\bm\theta}^{(q)}$ satisfies the linear constraints
$
{\bm A} {\bm p}_{\bms{\hat\theta}^{(q)}}=\hat{\bm C}.
$
Let us consider the asymptotic distribution of  $\hat{\bm\theta}^{(q)}$ under an assumption: the observed frequency vector $\hat {\bm\pi}$ is randomly sampled from a categorical distribution with a frequency vector $\bm\pi$ with size $N$.
Then, $\sqrt N(\hat{\bm\pi}-{\bm\pi})$ asymptotically converges to a normal distribution with mean $\bf 0$ and covariance $\bm G_{\bms\pi}^{-1}$,
where $\bm G_{\bms\pi}$ is the Fisher-Rao information matrix  as defined in \eqref{FR}.
Hence, $\sqrt N(\hat{\bm\theta}^{(q)}-{\bm\theta})$ converges to a normal distribution with mean $\bf0$ and covariance
$$
({\bm A} \bm J_{\bms\theta})^{-1}{\bm A}G_{\bms\pi} ^{-1}{\bm A}^\top (\bm J_{\bms\theta}^\top {\bm A}^\top)^{-1},
$$
where ${\bm J}_{\bms\theta}$ is the Jacobian matrix of the model ${\bm p}_{\bms\theta}^{(q)}$, or 
$$
{\bm J}_{{\bms\theta}} := \nabla_{\bms\theta}{\bm p_{{\bms\theta}}^{(q)}}=
\Big({\rm diag}\big( \bm p_{\bms\theta}\odot ({\bf1}+\bm A^{\top})^{-1}\big)-
\frac{{\bf1}^\top({\bf1}+\bm A^{\top})^\frac{2-q}{q-1}}{{\bf1}^\top({\bf1}+\bm A^{\top})^\frac{1}{q-1}}{\rm diag}\big(\bm p_{\bms\theta}\big)\Big)\bm A^\top.
$$
Finally, we note that, if $q=1$, the Hill number of order $q$ become equivalent to the Boltzmann-Shannon entropy, and
the minimum $q$-order estimator is reduced the maximum likelihood estimator.

This section demonstrates how the framework of information geometry, particularly through the concept of $q$-geodesics, enriches our understanding of diversity measures on the simplex $\Delta_{S-1}$. 
By connecting statistical divergences, geodesic paths, and maximum diversity distributions, we establish a geometric interpretation of the Hill numbers and related diversity indices. 
The foliation of the simplex into hyperplanes corresponding to linear constraints further illustrates the interplay between geometric structures and ecological considerations. 
These insights not only deepen the theoretical foundations but also provide practical tools for analyzing and optimizing biodiversity under various constraints.

%Let us consider an extension of the geometry on $\Delta_{S-1}$ to that on $\mathbb R^S_+$, where $\mathbb R_+=[0,\infty)$.  Thus, an element $\bm x=(x_i)$ of $\mathbb R_+^S$ has all nonnegative components $x_i$. All  geodesics defined on $\Delta_{S-1}$ can be extended to those on $\mathbb R_+^S$, for example, the mixture geodesic: $\bm x^{(\rm m)}_t=(1-t)\bm x+t\bm y$;the exponential geodesic: $\bm x^{(\rm e)}_t=\exp((1-t)\log(\bm x)+t\log(\bm y))$;the $q$ geodesic: $\bm x^{(q)}_t=((1-t)\bm x^{q-1}+t\bm y^{q-1})^\frac{1}{q-1}$.
%Consider a mapping $\Pi$ of $\mathbb R_+^S$ into $\Delta_{S-1}$. Then, for a curve $\{\bm x_t\}_{t\in[0,1]}$,   $\{\Pi(\bm x_t)\}_{t\in[0,1]}$ is a curve in $\Delta_{S-1}$. We consider the Hill number as ${}^q\!D_+(\bm x)=({\bf1}^\top\bm x^q)^\frac{1}{1-q}$, Rao's quadratic entropy ${\rm Q}_+(\bm x)=\bm x^\top \bm W \bm x.$

%%%%%%%%%%%%%%%%%%%%%%%%%%%%%%%%%%%%%%%%%%%%%%%%%%%%%%%%%%%%%%%%%%%%%%%%%%%%%%%%%%%%%%%%%%%%%%%%%%%%%%%%%%%%%%%%%%%%%%%%%%%%%%%%%%%%%%%%%%%%%%%%%%%%%%%%%%%%%%%%%%%%%%%%%%%%%%%%%%%%%%%%%%%%%%%%%%%%%%%%%%%%%%%%%%%%%%%%%%%%%%%%%%%%%%%%%%%%%%%%%%%%%%%%%

\section{Discussion}\label{sect5}
The maximum diversity distribution under a linear constraint embodies the interplay between species-specific ecological traits and the overall diversity of the community. 
It provides valuable insights into how ecological factors shape community composition and offers practical implications for biodiversity conservation and ecosystem management. 
By considering both the mathematical optimum and ecological principles, ecologists can better interpret diversity patterns, predict ecological dynamics, and set informed conservation goals that align with the inherent constraints of natural ecosystems.
By selecting different values of \( q \), ecologists can focus on different aspects of biodiversity (dominance or rarity), making this a powerful tool for assessing community dynamics, resource competition, and conservation strategies in response to environmental constraints.
For conservation applications, comparing observed abundance distribution to the theoretical maximal Hill number distributions offers a valuable diagnostic tool. 
Alignment with the maximal distribution suggests that the community may have reached an optimal diversity configuration for its specific ecological context, whereas deviations indicate potential imbalances or stresses, such as dominance by invasive species or suppression of keystone species. 
This comparison helps identify conservation needs and highlights areas for intervention, guiding efforts to maintain a balanced, resilient ecosystem. 
Monitoring these deviations over time can also track how community structure responds to environmental pressures, providing insights into resilience and adaptation.

Species Distribution Modeling (SDM) is a statistical approach used to predict the geographic distribution of species based on their known occurrences and environmental conditions. SDMs relate species presence-absence or presence-only data to environmental variables (e.g., temperature, precipitation, habitat features) to estimate the probability of species occurrence or abundance across spatial landscapes. These models are widely applied in ecology for biodiversity assessments, habitat suitability mapping, and predicting the impacts of environmental changes, such as climate change, on species distributions, see \cite{{southwood2009},{elith2009species},{Komori2019},{saigusa2024robust}}.
We discuss a potential framework and its implications for incorporating SDM into the approach of maximum diversity distributions represents.
SDMs naturally account for environmental covariates (temperature, precipitation, habitat types), providing realistic constraints for diversity maximization.
Using the predicted relative abundances from SDMs, the observed diversity can be directly compared to the maximum diversity, offering insights into the extent to which ecological constraints influence community assembly, see \cite{kusumoto2023}.
In areas with significant deviation from maximum diversity distributions, SDMs can provide insights into which species are missing or disproportionately abundant, suggesting specific restoration or management interventions.

Maximum diversity distributions offer a powerful framework that can extend beyond ecology to other domains like economics and engineering, providing insights into resource allocation, system robustness, and optimization under constraints.
In portfolio planning, maximum diversity principles can balance risk and return by distributing investments across assets.
In control problems, diversity-based strategies enhance robustness by distributing control efforts or resources across system components.
Applications include energy management in networks, feedback control in complex systems, and adaptive control for dynamic environments (e.g., autonomous systems, smart grids).
By adapting maximum diversity distributions to these fields,we can address challenges in resource optimization, system stability, and risk management.

\section*{Appendix}

We explore the ranges of $t$ for defining a mixture geodesic and $q$-geodesic in $\Delta_{S-1}$. 

\begin{proposition}\label{prop6}
	
	Consider a  mixture geodesic $\bm p^{\rm (m)}_t =(1-t)\bm \pi+t\bm p$ connecting  $\bm\pi$ and $\bm p$ in the interior of $\Delta_{S-1}$, where $0\leq t\leq 1$.  The range of $t$ is enlarged to a closed interval
	\[
	\left[ -\min\Big( \frac{1}{R_1-1},\frac{1}{Q_1-1}\Big),1+\min\Big( \frac{1}{R_2-1},\frac{1}{Q_2-1}\Big) \right] 
	\] 
	where
	\[
	R_1=\max_{i} \frac{p_i}{\pi_i}, \quad R_2=\max_{i} \frac{\pi_i}{p_i}, \quad
	Q_1=\max_{i} \frac{1-p_i}{1-\pi_i}, \quad Q_2=\max_{i} \frac{1-\pi_i}{1-p_i}.
	\]
\end{proposition}
\begin{proof}
	By definition,
	\begin{align}\label{ineq0}
		-\pi_i<(p_i-\pi_i)t<1-\pi_i\quad(i=1,...,S).
	\end{align}
	If $p_i>\pi_i$, then 
	\[
	-\frac{1}{\frac{p_i}{\pi_i}-1}<t<1+\frac{1}{\frac{1-\pi_i}{1-p_i}-1} \quad( i=1,...,S).
	\]
	This implies
	\begin{align}\label{ineq1}
		-\frac{1}{\max_i \frac{p_i}{\pi_i}-1}<t<1+\frac{1}{\max_i \frac{1-\pi_i}{1-p_i}-1}.
	\end{align}
	Similarly, if $p_i<\pi_i$, then
	\begin{align}\label{ineq2}
		-\frac{1}{\max_i \frac{1-p_i}{1-\pi_i}-1}<t<1+\frac{1}{\max_i \frac{\pi_i}{p_i}-1}.
	\end{align}
	Hence, this yields the closed interval combining inequalities \eqref{ineq1} and \eqref{ineq2}.
	
\end{proof}

We next focus on a case of the $q$-geodesic.
\begin{proposition}
	Consider a $q$-geodesic 
	$$
	\bm p^{(q)}_t =Z_t \{(1-t)\bm \pi^{q-1}+t\bm p^{q-1}\}^\frac{1}{q-1}
	$$ 
	connecting  $\bm\pi$ and $\bm p$ in in the interior of $\Delta_{S-1}$, where $0\leq t\leq 1$ and $z_t$ is the normalizing constant.  The range of $t$ is enlarged to a closed interval
	\begin{align}\label{CI2}
		\left[ -\min\Big( \frac{1}{R_1^{(q)}-1},\frac{1}{Q_1^{(q)}-1}\Big),1+\min\Big( \frac{1}{R_2^{(q)}-1},\frac{1}{Q_2^{(q)}-1}\Big) \right] 
	\end{align} 
	where
	\begin{align}\nonumber %\label{CI2}
		R_1^{(q)}=\max_{i} \frac{p_i^{q-1}}{\pi_i^{q-1}}, \quad R_2^{(q)}=\max_{i} \frac{\pi_i^{q-1}}{p_i^{q-1}}, \quad
		Q_1^{(q)}=\max_{i} \frac{1-(Z_tp_i)^{q-1}}{1-(Z_t\pi_i)^{q-1}}, \quad Q_2^{(q)}=\max_{i} \frac{1-(Z_t\pi_i)^{q-1}}{1-(Z_tp_i)^{q-1}}.
	\end{align}
\end{proposition}
\begin{proof}
	By definition,
	\[
	-(Z_t \pi_i)^{q-1}<\{(Z_t p_i)^{q-1}-(Z_t\pi_i^{q-1})\}t<1-(Z_t\pi_i)^{q-1}          \quad(i=1,...,S).
	\]
	which can be written as
	\[
	-\Pi_i<(P_i-\Pi_i\}t<1-\Pi_i          \quad(i=1,...,S).
	\]
	where $\Pi_i=(Z_t \pi_i)^{q-1}$ and $P_i=(Z_t p_i)^{q-1}$.
	This is essentially equal to \eqref{ineq0}.
	Therefore,  we conclude the closed interval \eqref{CI2} by the same discussion as that in the proof of Proposition \ref{prop6}.
\end{proof}

\end{document}